\documentclass[11pt]{article}

\usepackage[margin=1in]{geometry}
\usepackage{times}
\usepackage{amsfonts,amsmath,amssymb,mathtools}
\usepackage{amsthm}
\newtheorem{definition}{Definition}
\newtheorem{lemma}{Lemma}
\newtheorem{remark}{Remark}
\newtheorem{claim}{Claim}
\newtheorem{theorem}{Theorem}
\usepackage{microtype}
\usepackage{enumitem}
\usepackage{xcolor}
\usepackage{graphicx}
\usepackage{url}
\usepackage[numbers]{natbib}
\usepackage[colorlinks=true,linkcolor=blue,citecolor=blue,urlcolor=blue]{hyperref}

\title{Multiple Planted Structures Below $\sqrt{n}$:
An SoS Integrality Gap and an SQ Lower Bound}

\author{
Matvey Mosievskiy and Lev Reyzin\thanks{Supported in part by NSF grant ECCS-2217023}\\
Department of Mathematics, Statistics, \& Computer Science\\
University of Illinois Chicago\\
\small{\texttt{\{mmosie2,lreyzin\}@uic.edu}}
}

\date{}
\newcommand{\E}{\mathbb{E}}
\newcommand{\Prb}{\mathbb{P}}
\newcommand{\Var}{\operatorname{Var}}
\newcommand{\wtE}{\widetilde{\mathbb{E}}}
\newcommand{\cX}{\mathcal{X}}
\newcommand{\cP}{\mathcal{P}}
\newcommand{\T}{\mathcal{T}}
\newcommand{\supp}{\operatorname{supp}}
\newcommand{\Comp}{\operatorname{Comp}}
\newcommand{\sep}{\operatorname{sep}}

\begin{document}
\maketitle

\begin{abstract}
We study computational limitations in \emph{multi-plant} average-case inference problems, in which $t$ disjoint planted structures of size $k$ are embedded in a random background on $n$ elements.
A natural parameter in this setting is the total planted size $K := kt$.
For several classic planted-subgraph problems, including planted clique, existing algorithmic and lower-bound evidence suggests a characteristic computational threshold near $\sqrt{n}$ in the single-plant setting.

Our main result is a Sum-of-Squares (SoS) integrality gap for refuting the presence of multiple planted cliques.
Specifically, for $G \sim G(n,1/2)$, we construct a degree-$d$ SoS pseudoexpectation for the natural relaxation that maximizes the total size of up to $t$ disjoint cliques.
Throughout the regime
$
kt \le n^{1/2 - c\sqrt{d/\log n}},
$
for a universal constant $c>0$, this relaxation achieves objective value $kt(1-o(1))$, and therefore degree-$d$ SoS cannot certify an upper bound below $kt$.
This extends the planted-clique SoS lower bounds of~\cite{BarakHKKMP19} to a multi-plant setting with explicit disjointness constraints.
As complementary evidence from a different computational model, we prove a lower bound in the statistical query (SQ) framework, extending the results of~\cite{FeldmanGRVX17}.
We show that for detecting $t$ disjoint planted $k \times k$ bicliques (equivalently, a row-mixture distribution), when $kt = O(n^{1/2-\delta})$ for any fixed $\delta>0$, no polynomial-time SQ algorithm can distinguish the planted and null distributions with constant advantage.
\end{abstract}

\section{Introduction}
A broad family of average-case inference tasks can be phrased as distinguishing a
null random object from a planted alternative. Many such problems exhibit a
statistical--computational gap: even when detection is information-theoretically
possible, known polynomial-time algorithms fail below a characteristic signal
level.
A canonical example is the \emph{planted clique} problem~\citep{Jerrum1992,Kucera95},
where the goal is to detect a clique of size $k$ embedded in an
Erdős--Rényi random graph $G(n,1/2)$. Despite extensive research, no efficient
algorithm is known to detect cliques significantly smaller than $O(\sqrt{n})$,
and multiple computational frameworks provide evidence of hardness in this
regime~\citep{BarakHKKMP19,ChenMZ25,FeigeK00,FeldmanGRVX17,Jerrum1992,MekaPW15}.
The assumed hardness of planted clique has various implications, including for
cryptography~\citep{AbramBIKN23,ApplebaumBW10,JuelsP00}.

Motivated by this, we study \emph{multi-plant} settings, where the planted signal
is distributed across many disjoint structures. Related multi-clique formulations have appeared under the name \emph{planted clique cover}~\citep{Dughmi2014} and the \emph{stochastic block model}~\citep{HollandLL1983}.
In our model, there are $t$ disjoint planted
structures, each of ``size'' $k$, embedded into an ambient domain of size $n$.
We write $K := kt$ for the \emph{total planted size}. A recurring question is
whether planting many small structures can overcome the $\sqrt{n}$ barrier that
appears in the single-plant setting---i.e., whether the relevant computational
threshold is governed primarily by $K$.
We show that the multi-plant parameter $K=kt$ behaves like the single-plant size in terms of computational threshold lower bounds:
\begin{itemize}[leftmargin=2em,nosep]
  \item \textbf{Sum-of-Squares (SoS).} Our main result is an integrality gap for the SoS relaxation of refuting the presence of $t$ disjoint $k$-cliques in $G(n,1/2)$.
  \item \textbf{Statistical Query (SQ).} As a complementary result, we prove an SQ lower bound for detecting $t$ disjoint planted $k\times k$ bicliques in a bipartite graph, phrased as a distribution-testing problem.
\end{itemize}

Our lower bounds only address the regime $kt \le n^{1/2-\varepsilon}$.
In the regime $k \ll \sqrt{n} \ll kt$, we conjecture that the relevant detection threshold is
$
kt \approx {n}/{k}$
equivalently $t k^2 \approx n,$
as below this scale simple aggregate edge- and degree-based statistics have vanishing signal-to-noise.


\subsection{Main result (SoS): planted cliques in \texorpdfstring{$G(n,1/2)$}{G(n,1/2)}}

Our SoS result essentially states the following. There exist absolute constants $c_0,c_1>0$ such that for all sufficiently large $n$ and every degree parameter $d=d(n)$ in the range where the truncation parameter $\tau$ can be chosen as in Section~\ref{sec:sos_construction} (in particular, $d\le c_1\log n$), the degree-$d$ Sum-of-Squares (SoS) relaxation for the ``maximum total planted size'' program (defined in Section~\ref{sec:sos_program}) achieves value $\ge kt (1-o(1))$ throughout the regime
$
kt \;\le\; n^{1/2 - c_0\sqrt{d/\log n}}.
$
This result extends the planted-clique Sum-of-Squares integrality gaps of~\cite{BarakHKKMP19} and \cite{MekaPW15} to a multi-plant setting. The main technical novelty is enforcing disjointness across labels while maintaining a low-degree calibrated pseudoexpectation. As a consequence, throughout the stated regime, degree-$d$ SoS cannot certify that the maximum total size of $t$ disjoint cliques in $G\sim G(n,1/2)$ is smaller than $kt$; in particular, it cannot refute the existence of $t$ disjoint $k$-cliques.

This is meaningful as a ``refutation'' target because under the null $G\sim G(n,1/2)$ the largest clique has size $(2+o(1))\log_2 n$ w.h.p.; in particular, whenever $k\gg \log n$ the true combinatorial optimum (maximum total size of $t$ disjoint cliques) is w.h.p. strictly smaller than $kt$, even though the SoS relaxation still attains objective value $\approx kt$.

\subsection{Complementary result (SQ): planted \texorpdfstring{$k\times k$}{k x k} bicliques}

Our SQ lower bound is as follows.
Fix $\delta>0$. Consider the detection problem for $t$ disjoint planted $k\times k$ bicliques in a bipartite graph (equivalently, the row-mixture distributions $\{\mathsf D_S\}$ defined in Part~\ref{part:sq}). If $kt = O(n^{1/2-\delta})$, then no randomized SQ algorithm making $\mathrm{poly}(n)$ queries can distinguish $\mathsf D_S$ (for unknown planting $S$) from the null distribution with constant advantage over random guessing. Importantly, our SQ hard family ranges over \emph{all} ordered equal-size plantings; in particular, standard reductions from the single-plant setting do not directly yield such an equal-size multi-plant lower bound (see Remark~\ref{rem:partition}).
This is not merely a technicality: a na\"ive partition-based reduction would only suggest a computational threshold around
$k\lesssim \sqrt{n/t}$ (equivalently $kt\lesssim \sqrt{nt}$), whereas our direct construction reaches the substantially stronger barrier $kt\lesssim \sqrt{n}$.

The closest prior result to our SQ lower bound is due to ~\cite{KothariVWX23}, who showed that \emph{recovering} multiple planted cliques (or a full planted partition) is as hard as detecting a single planted clique.
They also show a detection--recovery gap: planting sufficiently many dense cliques can make \emph{detection} easy against random graphs even when recovery remains hard.
Thus, the implications of our results for recovery are covered by~\cite{KothariVWX23}, but not our main results on detection.
Moreover, their work applies to low-degree methods, whereas ours is in the SQ model; these models are nearly equivalent, but the equivalence requires additional assumptions~\citep{BrennanBHLS20}.

\section{A Sum-of-Squares Integrality Gap for Multiple Planted Cliques}\label{part:sos}

We follow the general calibration-and-truncation framework of
\cite{BarakHKKMP19,MekaPW15} for constructing Sum-of-Squares pseudoexpectations.
The main adaptation is that we work with double-indexed variables $x_{i,j}$ indicating
membership of vertex $i$ in planted clique $j$, together with hard disjointness
constraints across labels.
A key observation is that this multi-label structure does not complicate the PSD analysis:
inconsistent configurations (e.g., a vertex assigned to two labels) have vanishing 
planted expectations and thus vanishing calibrated pseudo-moments (Lemma~\ref{lem:coefbd}(1)),
and the coefficient bound $(kt/n)^{|V|}$ (Lemma~\ref{lem:coefbd}(2)) has the same form 
as in the single-plant case with $k$ replaced by the total planted size $kt$.
Once these facts are established, the ribbon-based separator and factorization arguments
from~\cite{BarakHKKMP19} carries over: every contributing $(A,B)$-ribbon admits a separator
of size at most $\min\{|\supp(A)|,|\supp(B)|\}$, and the terminal remainder in the
separator-layer decomposition vanishes (Lemma~\ref{lem:psd_core}).

\subsection{Planted cliques and SoS preliminaries}

\subsubsection{Model}

We consider the \emph{planted clique model} $\mathcal{G}(n, 1/2, k, t)$. A draw $(G,X)$ is generated as follows:
\begin{enumerate}[leftmargin=2em,nosep]
    \item Choose $t$ disjoint subsets $C_1,\dots,C_t\subseteq[n]$, each of size $k$, uniformly at random.
    \item Let $H$ be the graph with edges exactly within each $C_j$ (so each $C_j$ induces a clique).
    \item Sample $G_0\sim G(n,1/2)$.
    \item Output $G = G_0 \cup H$.
    \item Let $X=(x_{i,j})_{i\in[n],j\in[t]}\in\{0,1\}^{n\times t}$ be the indicator matrix with $x_{i,j}=1$ iff $i\in C_j$.
\end{enumerate}
Then $\sum_{j=1}^t x_{i,j}\le 1$ for all $i$, and $\sum_{i=1}^n x_{i,j}=k$ for all $j$.
We work with Boolean variables $\{x_{i,j}\}_{(i,j)\in[n]\times[t]}$ and identify polynomials with their multilinear representatives (using $x_{i,j}^2=x_{i,j}$). Let $\cP_{n,t}^d$ be the space of multilinear polynomials of degree at most $d$.

\begin{definition}[degree-$d$ pseudoexpectation]\label{def:pseudoexp}
A linear operator $\wtE:\cP_{n,t}^d\to\mathbb{R}$ is a degree-$d$ pseudoexpectation if
\begin{enumerate}[leftmargin=2em,nosep]
\item Normalization: $\wtE[1]=1$.
\item PSD: $\wtE[p^2]\ge 0$ for every polynomial $p\in \cP_{n,t}^{\lfloor d/2\rfloor}$.
\end{enumerate}
\end{definition}

\begin{definition}[constraints]
A degree-$d$ pseudoexpectation $\wtE$ \emph{satisfies} a constraint $q=0$ (where $q\in\cP_{n,t}^d$) if for every polynomial $r$ such that $\deg(qr)\le d$, we have $\wtE[qr]=0$.
\end{definition}

\subsubsection{Planted-clique constraints and the SoS refutation program}\label{sec:sos_program}
We encode ``$t$ disjoint cliques'' using only hard zero-constraints; the size is enforced via an objective. (Booleanity $x_{i,j}^2=x_{i,j}$ is enforced implicitly by working in the Boolean quotient described above.)
\begin{enumerate}[leftmargin=2em,nosep]
\item \textbf{Clique (non-edges):} for every non-edge $\{u,v\}\notin E(G)$ and each $r\in[t]$, $x_{u,r}x_{v,r}=0.$
\item \textbf{Disjointness:} for all $i\in[n]$ and $j_1\ne j_2$, $x_{i,j_1}x_{i,j_2}=0.$
\end{enumerate}
The associated degree-$d$ SoS relaxation is the optimization problem
\vskip-3pt
\begin{equation}\label{sosopt}
\mathrm{OPT}_{d}(G):=\max \; \wtE\Big[\textstyle\sum_{j=1}^t\textstyle\sum_{i=1}^n x_{i,j}\Big]
\end{equation}
\vskip-3pt
\noindent over all degree-$d$ pseudoexpectations $\wtE$ that satisfy the two constraint families above.
We say degree-$d$ SoS ``refutes'' $t$ disjoint $k$-cliques if it certifies $\mathrm{OPT}_{d}(G)<kt$.

The above program computes the maximum \emph{total} size of up to $t$ disjoint cliques in $G$.
The relaxation does not enforce equal sizes across labels; some labels may be unused.
Nevertheless, certifying $\mathrm{OPT}_d(G) < kt$ implies that $G$ does not contain
$t$ disjoint cliques each of size $k$.
Our integrality gap therefore rules out even this weaker form of refutation.

\subsection{Construction of the pseudoexpectation via truncated Fourier expansion}\label{sec:sos_construction}

We follow the now-standard approach for planted-subgraph SoS lower bounds~\citep{MekaPW15,BarakHKKMP19}: define pseudo-moments by \emph{calibration} (matching planted expectations) and then truncate in a \emph{local} manner so that (i) the resulting operator is well-defined as a low-degree polynomial in $G$, and (ii) the corresponding moment matrix can be shown PSD via a ribbon argument.

\subsubsection{Fourier basis on \texorpdfstring{$G(n,1/2)$}{G(n,1/2)}}

Encode each edge variable as $G_e\in\{\pm 1\}$ with $G_e=1$ if $e\in E(G)$ and $G_e=-1$ otherwise. For $T\subseteq\binom{[n]}{2}$ define $\chi_T(G)=\prod_{e\in T}G_e$. The characters $\{\chi_T\}$ form an orthonormal basis for functions of $G$ under the inner product $\langle f,h\rangle = \E_{G\sim G(n,1/2)}[f(G)h(G)]$.

\subsubsection{Parameters}

Fix $\varepsilon\in(0,1/2)$ with
$
kt = n^{1/2-\varepsilon}.
$
Let $d=d(n)$ be the SoS degree parameter and let $\tau=\tau(n)$ be a truncation parameter.
We assume that $\varepsilon$, $d$, and $\tau$ satisfy
\begin{equation}\label{eq:param_assumptions}
Ct\tfrac{d}{\varepsilon}\le \tau \le \tfrac{\varepsilon}{C}\log n,
\qquad
\varepsilon > C\tfrac{\log\log n}{\log n},
\end{equation}
for a sufficiently large absolute constant $C>0$.
Under these conditions, all truncation, concentration, and PSD arguments
used in Sections~\ref{sec:sos_construction}--\ref{sec:sos_psd} apply.
Eventually one may take $d=\Theta(\varepsilon^2\log n)$ and $\tau=\Theta(d/\varepsilon)$, which satisfies Problem~\eqref{eq:param_assumptions}.

\subsubsection{The truncated operator}

For a set $M\subseteq [n]\times[t]$ with $|M|\le d$ define the monomial
\[
X_M := \textstyle\prod_{(i,j)\in M} x_{i,j}.
\]
Write $S_M:=\{ i\in[n]: \exists j \text{ with } (i,j)\in M\}$ for the set of vertices appearing in $M$.
For each label $r\in[t]$, define $S_r(M):=\{i\in[n]:(i,r)\in M\}$ and let
\[
\mathcal{E}_M := \textstyle\bigcup_{r=1}^t \textstyle\binom{S_r(M)}{2}
\]
be the set of (intra-label) edges \emph{implied} by $M$.
Given an edge set $T\subseteq\binom{[n]}{2}$, define the \emph{test graph}
\[
G_{M,T} := (V,\,T\cup \mathcal{E}_M),\qquad V := V(T)\cup S_M,
\]
and let $U_1,\dots,U_c$ be its connected components.
Let $\Comp(G_{M,T})$ denote the set of connected components of $G_{M,T}$.

\begin{definition}[truncated coefficients]\label{def:trunc}
Let $M\subseteq[n]\times[t]$ with $|M|\le d$ and let $T\subseteq\binom{[n]}{2}$. Define
\[
\widehat{\wtE[X_M]}(T) :=
\begin{cases}
\E_{(G,X)\sim \mathcal{G}(n,1/2,k,t)}\big[\chi_T(G)\,X_M\big],
& \text{if } |V|\le \tau \text{ and } |U_i|\le k \text{ for all }i,\\
0,&\text{otherwise.}
\end{cases}
\]
And 
define
$
\wtE_G[X_M] := \sum_{T\subseteq\binom{[n]}{2}} \widehat{\wtE[X_M]}(T)\,\chi_T(G).
$
\footnote{Note that by definition $\widehat{\wtE[X_M]}(T)=0$ unless $|V(T)\cup S_M|\le \tau$
and every connected component of $G_{M,T}=(V(T)\cup S_M,\,T\cup \mathcal{E}_M)$ has size $\le k$.}
\end{definition}
We note that if for some $(M,T)$ the planted expectation
$
\E_{(G,X)\sim\mathcal G(n,1/2,k,t)}[\chi_T(G)X_M]=0,
$
then the corresponding truncated Fourier coefficient
$\widehat{\wtE[X_M]}(T)=0$, regardless of truncation.
Thus truncation can only eliminate nonzero coefficients; it never introduces new ones.

\subsection{Key lemmas: coefficients, calibration, constraints}\label{sec:sos_key_lemmas}

\subsubsection{A bound on planted Fourier coefficients}

\begin{definition}[consistency of $(M,T)$]\label{def:consistency}
A pair $(M,T)$ is \emph{consistent} if there exists a valid assignment of vertices to labels $1,\dots,t$ consistent with:
(i) each $(i,j)\in M$ constrains vertex $i$ to lie in clique (label) $j$, and
(ii) every edge in $T$ connects two vertices that share a label.
If no such assignment exists, $(M,T)$ is \emph{inconsistent}.
\end{definition}

The bound in the following lemma reflects that every vertex appearing in $(M,T)$ must lie in one of the $kt$
planted positions. Truncation ensures that no large combinatorial amplification occurs.

\begin{lemma}\label{lem:coefbd}
Let $M\subseteq[n]\times[t]$ with $|M|\le d$ and let $T\subseteq\binom{[n]}{2}$. Let $V=V(T)\cup S_M$ and let $U_1,\dots,U_c$ be the connected components of the test graph $G_{M,T}=(V,T\cup \mathcal{E}_M)$.
\begin{enumerate}[leftmargin=2em, nosep]
\item If $(M,T)$ is inconsistent, then $\E_{(G,X)\sim \mathcal{G}(n,1/2,k,t)}[\chi_T(G)\,X_M]=0$.
\item If $(M,T)$ is consistent, then
\vskip-.5em
\[
\Big|\E_{(G,X)\sim \mathcal{G}(n,1/2,k,t)}[\chi_T(G)\,X_M]\Big|
\le t^c \Bigl(\frac{k}{n-|V|+1}\Bigr)^{|V|}
\le \Bigl(\frac{kt}{n-|V|+1}\Bigr)^{|V|}.
\]
\noindent In particular, if $|V|\le \tau$, then
\vskip-.5em
\[
\Big|\E_{(G,X)\sim \mathcal{G}(n,1/2,k,t)}[\chi_T(G)\,X_M]\Big|
\le t^c \Bigl(\frac{k}{n-\tau}\Bigr)^{|V|}
\le \Bigl(\frac{kt}{n-\tau}\Bigr)^{|V|}.
\]

\end{enumerate}
\end{lemma}

\begin{proof}
We use the law of total expectation, conditioning on the planted cliques $X$ (equivalently, on $(C_1,\dots,C_t)$).
For a fixed planting $X$, let
$
F_X := \bigcup_{r=1}^t \binom{C_r}{2}
$
be the set of edges forced to be present by the planted cliques. Recall that $G=G_0\cup H$, where $G_0\sim G(n,1/2)$
has independent unbiased $\{\pm1\}$ edge variables.

For any fixed $X$, we have
$
\E_{G\mid X}\!\left[\chi_T(G)\right]
=
1 \text{if } T\subseteq F_X$ and $0$ otherwise.
Indeed, if $T\subseteq F_X$, then all edges in $T$ are deterministically present, so each $G_e=+1$ and $\chi_T(G)=1$.
If $T\not\subseteq F_X$, then $T$ contains some edge $e$ not forced by $H$; this $G_e$ is an independent unbiased background
edge with $\E[G_e\mid X]=0$, so $\E[\chi_T(G)\mid X]=0$.
Therefore,
$\E_{(G,X)\sim \mathcal{G}(n,1/2,k,t)}\!\left[\chi_T(G)\,X_M\right]
= \E_X\!\left[X_M\cdot \E_{G\mid X}\!\left[\chi_T(G)\right]\right]
= \Pr_X\!\left[X_M=1 \ \wedge\  T\subseteq F_X\right].
$

If $(M,T)$ is inconsistent (Definition~\ref{def:consistency}), then there is \emph{no} planting $X$ satisfying $X_M=1$
for which every edge of $T$ lies within a single planted clique (equivalently, for which $T\subseteq F_X$).
Hence $\Pr_X[X_M=1 \wedge T\subseteq F_X]=0$, and the planted Fourier coefficient is $0$.

So, we assume $(M,T)$ is consistent and let $V=V(T)\cup S_M$. Let $U_1,\dots,U_c$ be the connected components of the test graph
$G_{M,T}=(V,\,T\cup \mathcal{E}_M)$.
In any planting $X$ contributing to the event $\{X_M=1 \wedge T\subseteq F_X\}$, every edge in $T\cup\mathcal{E}_M$
must lie within a planted clique, hence each component $U_\ell$ must be assigned a single label in $[t]$.
(If some $|U_\ell|>k$, this event is impossible.)

Fix any assignment $\sigma:\Comp(G_{M,T})\to[t]$ of labels to components that is consistent with the pinned variables in $M$.
For each $r\in[t]$, let $V_r:=\bigcup_{\ell:\,\sigma(U_\ell)=r} U_\ell$ and set $s_r:=|V_r|$, so $\sum_r s_r=|V|$.
Since $(C_1,\dots,C_t)$ are sampled as ordered disjoint $k$-subsets, a sequential exposure/counting argument gives
\[
\Prb\bigl[\forall r\in[t],~ V_r\subseteq C_r \bigr]
~=~ \frac{\prod_{r=1}^t (k)_{s_r}}{(n)_{|V|}}
~\le~ \frac{\prod_r k^{s_r}}{(n)_{|V|}}
~=~ \frac{k^{|V|}}{(n)_{|V|}}
~\le~ \Bigl(\frac{k}{n-|V|+1}\Bigr)^{|V|}.
\]
There are at most $t^c$ choices of $\sigma$ (one label per component), hence
$
\Prb\!\left[X_M=1 \wedge T\subseteq F_X\right]
\le t^c\Bigl(\frac{k}{n-|V|+1}\Bigr)^{|V|}.
$
This proves the first bound in (2), and the displayed bounds for the case $|V|\le \tau$ follow by replacing $|V|$ with $\tau$
in the denominator and using $c\le |V|$ to get
$t^c k^{|V|}\le (kt)^{|V|}$.
\end{proof}

\subsubsection{Calibration}

The next lemma formalizes that, for coefficients supported on ``local'' Fourier characters that pass the truncation filter, the pseudoexpectation matches the planted expectation on average over the null.

\begin{lemma}\label{lem:calibration_sos}
Let $f_G(X)=\sum_{|M|\le d} c_M(G)\,X_M$ be a polynomial in $\{x_{i,j}\}$ of degree at most $d$. Assume each coefficient function $c_M(G)$ has a Fourier expansion $c_M(G)=\sum_T \widehat{c_M}(T)\chi_T(G)$ supported only on sets $T$ satisfying
$
|V(T)\cup S_M|\le \tau$
and every component of $(V(T)\cup S_M,\,T\cup \mathcal{E}_M)$ has size $\le k.$
Then
$
\E_{G\sim G(n,1/2)}\big[\wtE_G f_G(X)\big]
\;=\;
\E_{(G,X)\sim \mathcal{G}(n,1/2,k,t)}\big[f_G(X)\big].
$
\end{lemma}

\begin{proof}
Expand $c_M(G)$ and $\wtE_G[X_M]$ in the orthonormal Fourier basis and use Parseval/orthogonality:
only matching characters survive under $\E_{G\sim G(n,1/2)}[\chi_T(G)\chi_{T'}(G)]$.
Under the stated support condition, $\widehat{\wtE[X_M]}(T)=\E_{(G,X)}[\chi_T(G)X_M],$ and the claim follows by linearity.
\end{proof}

\subsubsection{Satisfaction of hard constraints}

\begin{lemma}[proof in~\ref{app:sos_constraints}]
 \label{lem:nonedge}
For any graph $G$ and any non-edge $\{u,v\}\notin E(G)$, any $r\in[t]$, and any polynomial $p$ with $\deg(p)\le d-2$,
$
\wtE_G[x_{u,r}x_{v,r}p]=0.
$
\end{lemma}

\begin{lemma}\label{lem:disjoint}
For any graph $G$, any $i\in[n]$, any distinct labels $j_1\ne j_2$, and any polynomial $p$ with $\deg(p)\le d-2$,
$
\wtE_G[x_{i,j_1}x_{i,j_2}p]=0.
$
\end{lemma}
\begin{proof}
For any monomial $p=X_{M'}$, the planted expectation $\E_{(G,X)}[\chi_T(G)\,x_{i,j_1}x_{i,j_2}X_{M'}]$ is $0$
for every $T$ because in the planted model a vertex cannot belong to two disjoint planted cliques.
Hence every Fourier coefficient of $\wtE_G[x_{i,j_1}x_{i,j_2}X_{M'}]$ is $0$ and so is the pseudoexpectation.
\end{proof}

\subsubsection{Size and normalization}

The operator $\wtE_G$ is calibrated so that under the null average over $G\sim G(n,1/2)$ it matches the planted expectations of $1$ and of the size statistics $\sum_i x_{i,j}$. The truncation step introduces a small (but nonzero) error, so we record the resulting low-degree approximate satisfaction of the size constraints.

Note that the polynomials used for normalization and size, namely $f=1$ and
$f=\sum_{i=1}^n x_{i,j}$, have coefficients $c_M(G)$ that are constant functions of $G$.
Thus they are Fourier-supported only on $T=\emptyset$, and the support condition in
Lemma~\ref{lem:calibration_sos} holds trivially for these choices of $f_G(X)$.

\begin{lemma}[proof in~\ref{app:sos_constraints}]\label{lem:concentration}
With high probability,
$
\wtE_G[1]=1\pm n^{-\Omega(\varepsilon)}$
and also, for all $j\in[t]$,
$\wtE_G\Big[\sum_{i=1}^n x_{i,j}\Big]=k(1\pm n^{-\Omega(\varepsilon)}).
$
\end{lemma}

\begin{remark}
One can enforce exact normalization by defining $\wtE_G^\star[p]=\wtE_G[p]/\wtE_G[1]$.
This preserves all hard zero-constraints (since those expectations are exactly $0$). Moreover,
Lemma~\ref{lem:concentration} implies that $\wtE_G^\star$ achieves objective value
$\wtE_G^\star[\sum_{j,i}x_{i,j}]=kt\cdot (1\pm n^{-\Omega(\varepsilon)})$.
\end{remark}

\subsection{PSD via moment matrix and ribbon factorization}\label{sec:sos_psd}

\subsubsection{Moment matrix and PSD criterion}

\begin{definition}[moment matrix]\label{def:moment_matrix}
Let $\wtE_G$ be a degree-$d$ pseudoexpectation candidate. The \emph{moment matrix} $M$ is indexed by subsets
$A,B\subseteq \cX=[n]\times[t]$ of size at most $\lfloor d/2\rfloor$:
\vskip-.5em
\[
M(A,B):=\wtE_G[X_A X_B],
\qquad
X_A:=\textstyle\prod_{(i,j)\in A} x_{i,j}.
\]
\end{definition}
\vskip-.5em
By Definition~\ref{def:pseudoexp}, $\wtE_G$ is PSD iff $M\succeq 0$.

\subsubsection{Ribbons: expansion of entries as low-degree Fourier sums}

We next express each entry $M(A,B)$ as a sum over small edge-sets (Fourier characters) that pass truncation.
Write $\supp(A)\subseteq[n]$ for the set of vertices appearing in $A$ (i.e., $\supp(A)=S_A$).

For notational clarity, for $A,B\subseteq \cX$ define the associated pinned-edge set
\vskip-1em
\[
\mathcal{E}_{A,B} := \mathcal{E}_A \cup \mathcal{E}_B
\qquad\text{where}\qquad
\mathcal{E}_A=\textstyle\bigcup_{r=1}^t \textstyle\binom{S_r(A)}{2}.
\]
\vskip-1pt
\noindent Given an edge-set $W\subseteq \binom{[n]}{2}$ define the test graph
\[
G_{A,B,W} := \bigl(V,\, W\cup \mathcal{E}_{A,B}\bigr),\qquad V:=V(W)\cup \supp(A)\cup\supp(B),
\]
and let $\Comp(A,B,W)$ denote the set of connected components of $G_{A,B,W}$.

\begin{definition}[$(U,V;K)$-ribbon]\label{def:uvk_ribbon}
For an edge set $F\subseteq\binom{[n]}{2}$, let $V(F)$ denote its set of endpoints.
Fix vertex sets $U,V\subseteq[n]$ and a pinned-edge set $K\subseteq\binom{[n]}{2}$.
A \emph{$(U,V;K)$-ribbon} is a graph
\vskip-.5em
\[
R=(V(R),E(R))
\quad\text{with}\quad
U\cup V\subseteq V(R),\ \ K\subseteq E(R),\ \ E(R)\subseteq \textstyle\binom{V(R)}{2},
\]
\vskip-1pt
\noindent such that every vertex of $R$ lies in $V(E(R))\cup U\cup V$.
We write $W_R:=E(R)\setminus K$ for the \emph{free} (non-pinned) edges of the ribbon.
\end{definition}

\begin{definition}[$(A,B)$-ribbon]\label{def:ab_ribbon}
For $A,B\subseteq \cX=[n]\times[t]$, an \emph{$(A,B)$-ribbon} is a
$(\supp(A),\supp(B);\mathcal{E}_{A,B})$-ribbon, where $\mathcal{E}_{A,B}:=\mathcal{E}_A\cup \mathcal{E}_B$ and
$\mathcal{E}_A=\bigcup_{r=1}^t \binom{S_r(A)}{2}$.
\end{definition}
Although ribbons are graphs on the vertex set $[n]$, the label assignments in
$A,B\subseteq[n]\times[t]$ enter through the pinned-edge set
$\mathcal{E}_{A,B}=\mathcal{E}_A\cup\mathcal{E}_B$ (which depends on the label-partitions $S_r(A),S_r(B)$)
and through the consistency condition for triples $(A,B,W)$ used in Lemma~\ref{lem:ribbon_expansion}.

\begin{lemma}\label{lem:ribbon_expansion}
For $A,B$ with $|A|,|B|\le \lfloor d/2\rfloor$,
$
M(A,B)=\wtE_G[X_A X_B]
=\sum_{W\subseteq\binom{[n]}{2}} \alpha_{A,B}(W)\,\chi_W(G),
$
where $\alpha_{A,B}(W)=0$ unless the $(A,B)$-ribbon $R$ induced by $W$ is truncation-feasible.
Moreover, whenever $(A,B,W)$ is consistent with the pinning constraints induced by $A,B$ and the edge-set $W$,
\[
|\alpha_{A,B}(W)|
\le
t^{\#\Comp(A,B,W)}
\left(\tfrac{k}{n-|V(R)|+1}\right)^{|V(R)|}
\le
t^{\#\Comp(A,B,W)}
\left(\tfrac{k}{n-\tau}\right)^{|V(R)|}
\le
\left(\tfrac{kt}{n-\tau}\right)^{|V(R)|}.
\]
\end{lemma}

\begin{proof}
This is the specialization of Definition~\ref{def:trunc} to the monomial $X_A X_B$ and Fourier character $\chi_W$,
together with the planted Fourier coefficient bound (Lemma~\ref{lem:coefbd}) applied to the test graph
$G_{A,B,W}=(V,W\cup \mathcal{E}_{A,B})$.
\end{proof}

\subsubsection{Separators and canonical factorization}

We now apply the ribbon separator framework to factor contributions by a minimum vertex separator.

\begin{definition}[vertex separators for $(A,B)$-ribbons]\label{def:sep_ab}
Let $R$ be an $(A,B)$-ribbon with vertex set $V(R)$ and edge set $W_R$. A set $Q\subseteq V(R)$ (not necessarily disjoint from $\supp(A)\cup\supp(B)$) is a \emph{separator}
if removing $Q$ disconnects $\supp(A)$ from $\supp(B)$ in the graph $(V(R),W_R\cup \mathcal{E}_{A,B})$.
Let $\sep(R)$ be the minimum size of a separator.\footnote{By convention, separators are allowed to intersect the endpoint sets $\supp(A)$ and $\supp(B)$.
In particular, $Q=\supp(A)$ is always a valid separator (removing $\supp(A)$ vacuously disconnects $\supp(A)$ from $\supp(B)$),
so $\sep(R)\le |\supp(A)|$ for every $(A,B)$-ribbon $R$.}
\end{definition}

We allow separators to intersect $\supp(A)\cup\supp(B)$.
With this convention, $\supp(A)$ and $\supp(B)$ are always valid separators, so for every $(A,B)$-ribbon $R$,
$
\sep(R)\le \min\{|\supp(A)|,|\supp(B)|\}.
$
This is convenient for the iterated factorization: when we iterate up to layer $m=\lfloor d/2\rfloor$ (the largest possible
index size in the moment matrix), there are no remaining ribbons with separator size $>m$.

\begin{lemma}[proof in~\ref{app:sos_separators}]\label{lem:left_right_sep_full}
Let $R$ be an $(A,B)$-ribbon. There exist unique minimum separators $S_L(R)$ and $S_R(R)$, called
the leftmost and rightmost minimum separators, respectively.
\end{lemma}

\begin{definition}[canonical factorization]\label{def:canonical_factorization_full}
Let $R$ be an $(A,B)$-ribbon and let $S_L=S_L(R)$ and $S_R=S_R(R)$. Define $(R_\ell,R_m,R_r)$ as the \emph{canonical factorization} of $R$:
\begin{itemize}[leftmargin=2em,nosep]
\item $R_\ell$ as the induced subgraph on vertices reachable from $\supp(A)$ without passing through $S_L$,
together with $S_L$; it is an $(\supp(A),S_L)$-ribbon.
\item $R_r$ analogously from $\supp(B)$ to $S_R$; it is an $(S_R,\supp(B))$-ribbon.
\item $R_m$ as the induced subgraph on vertices between $S_L$ and $S_R$; it is an $(S_L,S_R)$-ribbon.
\end{itemize}
\end{definition}

\begin{claim}\label{clm:rib_add_full}
Given $(R_\ell,R_m,R_r)$,
$|V(R)|=|V(R_\ell)|+|V(R_m)|+|V(R_r)|-|S_L(R)|-|S_R(R)|.
$
\end{claim}

\begin{proof}
Each vertex of $R$ lies in exactly one of the three regions, except vertices in $S_L$ (counted in $R_\ell$ and $R_m$)
and vertices in $S_R$ (counted in $R_m$ and $R_r$).
\end{proof}

\subsubsection{Matrix factorization: \texorpdfstring{$M = L Q L^\top - E$}{M = L Q L transpose - E}}

The key PSD strategy is to rewrite the ribbon sum so that the dominant term is a Gram matrix.
Let $\mathcal I$ be the set of all possible separators $S\subseteq[n]$ with $|S|\le \lfloor d/2\rfloor$.
We define:
\begin{itemize}[leftmargin=2em,nosep]
\item $L$ indexed by $A$ (rows) and $S\in\mathcal I$ (columns),
\item $Q$ indexed by $S,S'\in\mathcal I$.
\end{itemize}

Intuitively, $L(A,S)$ collects the contribution of left-ribbons from $\supp(A)$ to $S$, while $Q(S,S')$ collects the middle-ribbon
contribution between separators.

\begin{definition}[$L(A,S)$ and $Q(S,S')$]\label{def:LQ}
Define
\begin{align*}
L(A,S) &:= \textstyle \sum_{\substack{R:\,(\supp(A),S;\mathcal{E}_A)\text{-ribbon}\\ R\text{ feasible},\ \sep(R)=|S|}}
\beta_{A,S}(R)\,\chi_{W_R}(G),\ \text{and} \\
Q(S,S') &:=  \textstyle \sum_{\substack{R:\,(S,S';\emptyset)\text{-ribbon}\\ R\text{ feasible},\ \sep(R)=\min(|S|,|S'|)}}
\gamma_{S,S'}(R)\,\chi_{W_R}(G).
\end{align*}
where the coefficients $\beta,\gamma$ are the planted-calibrated Fourier coefficients arising from the truncation rule
(Definition~\ref{def:trunc}) specialized to these endpoint sets.
\end{definition}

Formally, the coefficients $\beta_{A,S}(R)$ and $\gamma_{S,S'}(R)$ are obtained by taking the Fourier coefficient
$\alpha_{A,B}(W_R)$ from Lemma~\ref{lem:ribbon_expansion} for the appropriate endpoint sets and then restricting to ribbons
whose canonical minimum separator is exactly the designated separator set (to avoid overcounting).

The restriction ``$\sep(R)=|S|$'' (resp. the analogous condition for $Q$) is the standard device that prevents overcounting
and ensures the separating factorization is well-defined.

\begin{lemma}[proof in~\ref{app:sos_factorization}]\label{lem:sep_factor}
There is a decomposition
$
M = L Q L^\top - E_1,
$
where the error matrix $E_1$ is supported on ribbons $R$ whose canonical separators satisfy $S_L(R)\neq S_R(R)$
(i.e., the factorization introduces overlaps due to non-disjoint separator regions).\footnote{The treatment of overlap terms with $S_L(R)\neq S_R(R)$ follows the iterated separator-layer factorization of~\cite[Section~6.2]{BarakHKKMP19}.
Although our ribbons include pinned edges $\mathcal E_{A,B}$ encoding labels, separators and canonical regions are defined purely via vertex cuts in the underlying graph, so the promotion argument carries over unchanged.}

\end{lemma}

\noindent The matrix $Q$ has the same ``middle-ribbon'' moment-matrix form as in prior work; we prove it is PSD with high probability.

\begin{lemma}[proof in~\ref{app:sos_factorization}]\label{lem:Q_psd}
W.h.p.\ over $G\sim G(n,1/2)$, the matrix $Q$ is positive semidefinite.
\end{lemma}


Each matrix $Q^{(s)}$ appearing in Lemma~\ref{lem:iter_sep} has the same ribbon-sum structure as $Q$ above, hence the same argument implies $Q^{(s)}\succeq 0$ with high probability. Since $s\le \lfloor d/2\rfloor=O(\log n)$, a union bound over all relevant $s$ preserves the ``with high probability'' guarantee.

\subsubsection{Iterated factorization and concluding PSD}

We iterate the one-step separating factorization to push contributions into increasing separator size.

\begin{lemma}[proof in~\ref{app:sos_factorization}]\label{lem:iter_sep}
There exist matrices $L^{(0)}:=L,L^{(1)},\dots,L^{(m)}$ and PSD matrices $Q^{(0)}:=Q,Q^{(1)},\dots,Q^{(m)}$ such that
$
M = \sum_{s=0}^m L^{(s)} Q^{(s)} (L^{(s)})^\top - E_{m+1},
$
where $E_{m+1}$ is defined to be the total contribution of truncation-feasible ribbons $R$ with $\sep(R)\ge m+1$ (equivalently, those not captured in the first $m+1$ separator layers).
In particular, $E_{m+1}$ is \emph{not} an ``overlap artifact'' remainder: the non-disjointness/overlap terms introduced by canonical factorization are absorbed into the successive $L^{(s)}Q^{(s)}(L^{(s)})^\top$ layers during the iteration.
\end{lemma}

The moment matrix always admits the decomposition
$
M = M^{(\mathrm{diag})} + M^{(\mathrm{off})},
$
where $M^{(\mathrm{diag})}$ is the diagonal matrix with entries $M^{(\mathrm{diag})}(A,A)=\wtE_G[X_A]$ and $M^{(\mathrm{off})}:=M-M^{(\mathrm{diag})}$ collects all off-diagonal entries.
So once Lemma~\ref{lem:psd_core} is established (so that $\wtE_G$ is a valid pseudoexpectation), the diagonal entries satisfy
$\wtE_G[X_A]=\wtE_G[X_A^2]\ge 0$ in the Boolean quotient, hence $M^{(\mathrm{diag})}\succeq 0$.

\begin{lemma}[proof in~\ref{app:sos_factorization}]\label{lem:psd_core}
Under the parameter setting of Section~\ref{sec:sos_construction}, with high probability over $G\sim G(n,1/2)$,
the moment matrix $M$ associated with $\wtE_G$ is positive semidefinite:
$
M \succeq 0.
$
\end{lemma}


\subsection{SoS lower bound}

\begin{theorem}[SoS lower bound]\label{thm:sos_formal}
There exists a constant $c_0>0$ such that if
$
kt \le n^{1/2 - c_0\sqrt{d/\log n}},
$
and $d$ lies in the range where the truncation parameter $\tau$ can be chosen as in Section~\ref{sec:sos_construction}, then with high probability over $G\sim G(n,1/2)$ the normalized functional $\wtE_G^\star[p]:=\wtE_G[p]/\wtE_G[1]$ is a degree-$d$ pseudoexpectation satisfying the hard constraints in Section~\ref{sec:sos_program} and achieving
\vskip-.7em
\[
\wtE_G^\star\Big[\textstyle\sum_{j=1}^t\textstyle\sum_{i=1}^n x_{i,j}\Big]\;\ge\; kt(1-o(1)).
\]
\end{theorem}

\begin{proof}
Construct $\wtE_G$ as in Definition~\ref{def:trunc} (working in the Boolean quotient so $x_{i,j}^2=x_{i,j}$ holds identically).
Lemmas~\ref{lem:nonedge} and \ref{lem:disjoint} imply exact satisfaction of the hard zero-constraints.
Lemma~\ref{lem:concentration} gives approximate normalization, which may be corrected by rescaling $\wtE_G^\star[p]=\wtE_G[p]/\wtE_G[1]$.
Lemma~\ref{lem:psd_core} implies PSD, hence $\wtE_G^\star$ is a valid degree-$d$ pseudoexpectation feasible for optimization program~\eqref{sosopt}.
Finally, Lemma~\ref{lem:concentration} implies $\wtE_G^\star[\sum_{j,i}x_{i,j}]=kt(1\pm n^{-\Omega(\varepsilon)})$, and so $\mathrm{OPT}_d(G)\ge kt(1-o(1))$.
\end{proof}
Consequently, degree-$d$ SoS fails to refute the existence of $t$ disjoint cliques of size $k$ in $G$.

\begin{remark}\label{rem:sos_no_reduction}
One might hope to derive Theorem~\ref{thm:sos_formal} from existing single-plant SoS lower bounds 
(e.g.,~\cite{MekaPW15,BarakHKKMP19}) via a black-box reduction. However, two obstacles prevent this.
First, the disjointness constraints in the multi-plant relaxation couple variables across cliques: 
a valid pseudoexpectation must satisfy $\wtE[x_{i,j_1} x_{i,j_2} p] = 0$ for all $i \in [n]$, 
all $j_1 \neq j_2$, and all polynomials $p$ of appropriate degree. 
This precludes na\"ive ``tensoring'' of $t$ independent single-plant pseudoexpectations, 
since such a product would not respect the constraint that each vertex belongs to at most one clique.
Second, a partition-based embedding running $t$ independent single-plant instances on disjoint 
vertex sets of size $n/t$ yields only the weaker threshold 
$
kt \le t  (n/t)^{1/2 - \varepsilon} = n^{1/2-\varepsilon}  t^{1/2+\varepsilon},
$
which degrades with $t$ rather than maintaining the sharp $kt \le n^{1/2-\varepsilon}$ barrier.
\end{remark}


\section{Statistical Query Lower Bound for Multiple Planted Bicliques}\label{part:sq}
Here we give a lower bound in the statistical query (SQ) model, introduced by~\cite{Kearns1998};
see also~\cite{BlumFJKMR94} for early foundational results and \cite{Reyzin2020} for a survey.
Our argument follows the statistical dimension framework of~\cite{FeldmanGRVX17},
but requires a nontrivial specialization to the multiple planted biclique (row-mixture)
setting.
In particular, our hard family consists of all ordered plantings of $t$ disjoint
$k$-subsets, and the analysis must control average correlations for every large
subfamily of these equal-size multi-plant distributions.

\subsection{Planted bicliques and SQ preliminaries}

\begin{definition}[planting]
A \emph{planting} is an ordered $t$-tuple $S=(S_1,...,S_t)$ of disjoint $k$-subsets~of~$[n]$.
\end{definition}

\begin{definition}[planted multi-biclique detection distribution]\label{def:Ds}
Let $n,k,t$ be positive integers with $kt\le n$, and let $S=(S_1,\dots,S_t)$ be a planting. Let
$\mathcal D_0=\mathrm{Uniform}(\{0,1\}^n)$ and $p=\frac{kt}{n}$. Define
\vskip-.8em
\[
  \mathsf D_S(x)
  \;=\;
  (1-p)\,\mathcal D_0(x)\;+\;\tfrac{p}{t}\textstyle\sum_{i=1}^t \mathcal D_0\!\left(x\,\middle|\, x_v=1~\forall v\in S_i\right).
\]
\end{definition}
\vskip-.5em
Equivalently, with probability $1-p$ one samples $x\sim\mathcal D_0$, and with probability $p$ one chooses
$i\in[t]$ uniformly, sets $x_v=1$ for all $v\in S_i$, and samples the remaining coordinates uniformly.
This is precisely a \emph{row-mixture} distribution, where each component corresponds to conditioning
$\mathcal D_0$ on a planted $k$-subset $S_i$.
The goal is to distinguish $\mathcal D_0$ from $\mathsf D_S$.

\begin{definition}[VSTAT oracle; c.f.~\cite{FeldmanGRVX17}]
The oracle $\mathrm{VSTAT}(N)$ for a distribution $D$ over domain $X$ takes a sample-size parameter $N$ and, for any query $f:X\to[0,1]$, returns a value $v$ satisfying
$
\bigl|v-\E_{x\sim D}[f(x)]\bigr|\le \max\!\left\{1/N,~\sqrt{{\mathrm{Var}_{x\sim D}[f(x)]}/{N}}\right\}.
$
\end{definition}

\begin{definition}[SDA; c.f.~\cite{FeldmanGRVX17}]\label{def:sda}
Let $\mathcal D_0$ be a reference distribution over $X$, and let $\mathcal F=\{D_1,\dots,D_m\}$ be a family of distributions over $X$.
Define $\widehat D_i(x)=D_i(x)/\mathcal D_0(x)-1$ and the inner product
$\langle f,g\rangle_{\mathcal D_0}=\E_{x\sim \mathcal D_0}[f(x)g(x)]$.
For $A\subseteq[m]$, set
$
\mathrm{AvgCorr}(A)=\frac{1}{|A|^2}\sum_{i,j\in A}\langle \widehat D_i,\widehat D_j\rangle_{\mathcal D_0}.
$
For $\bar\gamma>0$, $\mathrm{SDA}(\mathcal F,\mathcal D_0;\bar\gamma)$ is the largest $d$ s.t.\ every $A\subseteq[m]$ w/ $|A|\ge m/d$ has $\mathrm{AvgCorr}(A)\le\bar\gamma$.
\end{definition}

\begin{theorem}[\cite{FeldmanGRVX17}]\label{thm:sq_via_sda}
If $\mathrm{SDA}(\mathcal F,\mathcal D_0;\bar\gamma)\ge d$, then any randomized SQ algorithm that distinguishes a distribution in $\mathcal F$ from $\mathcal D_0$
with success probability at least $2/3$ requires at least $d$ queries to $\mathrm{VSTAT}(N)$, where $N=\Theta(1/\bar\gamma)$.
\end{theorem}

\subsection{Pairwise correlation}\label{sec:sq_pairwise}

Fix two plantings $S=(S_1,\dots,S_t)$ and $T=(T_1,\dots,T_t)$. Define
\vskip-.7em
\[
\widehat D_S(x)=\tfrac{\mathsf D_S(x)}{\mathcal D_0(x)}-1.
\]
Write $\lambda_{ij}=|S_i\cap T_j|$ and $\Lambda(S,T)=\sum_{i,j}\lambda_{ij}$.
\begin{lemma}[proof in~\ref{app:sq_proofs}]\label{lem:sq_corr}
For any two plantings $S,T$,
\vskip-.5em
\[
\bigl\langle \widehat D_S,\widehat D_T\bigr\rangle_{\mathcal D_0}
~=~p^2\!\left(\tfrac{1}{t^2}\textstyle\sum_{i,j}2^{\lambda_{ij}}-1\right)
~\le~ p^2\,\tfrac{2^{\Lambda(S,T)}-1}{t^2}
~\le~ \tfrac{k^2}{n^2}\,2^{\Lambda(S,T)}.
\]
\end{lemma}

\subsection{Counting overlaps}\label{sec:sq_count}

Let $m$ be the number of ordered plantings of $t$ disjoint $k$-subsets of $[n]$.

\begin{lemma}[proof in~\ref{app:sq_proofs}]\label{lem:sq_m}
If $kt\le n$, then
$
m=\binom{n}{kt}\,(kt)!/(k!)^t.
$
\end{lemma}

\begin{definition}[$\T_\ell$]
Fix a reference planting $S$. For $\ell\ge0$ let
$
\T_\ell=\{T:\Lambda(S,T)=\ell\}
$
be the collection of plantings with total overlap $\ell$ with $S$.
\end{definition}

\begin{lemma}\label{lem:sq_tl}
Fix integers $n,k,t$ and let $\ell=O(\ln n)$ with $\ell^2=o(kt)$. Assume $kt=O(n^{1/2-\delta})$ for some $\delta>0$. Then, as $n\to\infty$,
\vskip-1em
\[
{|\T_\ell|}/{m}
=\Prb[X=\ell]
={\textstyle\binom{kt}{\ell}\textstyle\binom{n-kt}{kt-\ell}}/{\textstyle\binom{n}{kt}}
=\tfrac{1}{\ell!}\Bigl(\tfrac{k^2t^2}{n}\Bigr)^{\!\ell}\,(1+o(1)).
\]
\end{lemma}

\begin{proof}
Since the $S_i$'s and $T_j$'s are pairwise disjoint within each family,
\(
\Lambda(S,T)=\bigl|\bigcup_i S_i \cap \bigcup_j T_j\bigr|.
\)
For uniformly random $T$, the union $\bigcup_j T_j$ is a uniformly random $kt$-subset of $[n]$.
Hence $X:=\Lambda(S,T)$ is hypergeometric with mass
$\Prb[X=\ell]={\binom{kt}{\ell}\binom{n-kt}{kt-\ell}}/{\binom{n}{kt}}$.
Since $kt=O(n^{1/2-\delta})$, we have $kt=o(\sqrt n)$.
Under the stated hypotheses $\ell=O(\ln n)$ and $\ell^2=o(kt)$, we may apply Lemma~\ref{lem:hypergeom_sparse_asymp} with $K=kt$ to obtain
$\Prb[X=\ell]
=\frac{1}{\ell!}\left({(kt)^2}/{n}\right)^{\!\ell}(1+o(1))
=\frac{1}{\ell!}\left({k^2t^2}/{n}\right)^{\!\ell}(1+o(1)).$
\end{proof}

\subsection{SQ lower bound}\label{sec:sq_main}

\begin{theorem}[proof in~\ref{app:sq_proofs}]\label{thm:sq_sda}
Suppose $kt=O(n^{1/2-\delta})$ for some fixed $\delta>0$.
Fix an integer $\ell$ such that
$
0\le \ell \le \min\left\{kt,\ \left\lfloor \log_2 ({n}/(k^2t^2))\right\rfloor-1\right\},
$
with $\ell = O(\log n)$ and $\ell^2=o(kt).$
Define
$
d:=\left\lfloor \ell!\left({n}/{(k^2t^2)}\right)^{\!\ell}\right\rfloor,$
$\bar\gamma:=2\,({k^2}/{n^2})\,2^\ell.
$
Then, for all sufficiently large $n$,
$
\mathrm{SDA}\bigl(\mathcal F,\mathcal D_0;\bar\gamma\bigr)\ge d,
$
where $\mathcal F=\{\mathsf D_S : S \text{ is an ordered planting}\}$.
\end{theorem}

\begin{proof}[of Theorem~\ref{thm:sq_sda}]
By Lemma~\ref{lem:sq_tl}, $m/d=|\T_\ell|(1+o(1))$. Let $A$ be any set of plantings with $|A|\ge m/d$.
Split pairs $(S,T)\in A\times A$ by $j=\Lambda(S,T)$.

\noindent \emph{Low overlap $j\le \ell$}:
By Lemma~\ref{lem:sq_corr}, each such pair contributes at most $(k^2/n^2)\,2^\ell$.

\noindent \emph{High overlap $j>\ell$}:
by Lemma~\ref{lem:sq_corr}, each pair with overlap $j$ contributes at most $(k^2/n^2)\,2^j$.
For fixed $S$, there are at most $|\T_j|$ choices of $T$ with overlap $j$, so the number of ordered pairs in $A\times A$ with overlap $j$
is at most $|A|\,|\T_j|$.
Thus the high-overlap contribution to average correlation is at most
\vskip-.5em
\[
\tfrac{k^2}{n^2}\textstyle\sum_{j>\ell} 2^j \cdot \frac{|A|\,|\T_j|}{|A|^2}
\le \frac{k^2}{n^2}\cdot \frac{d}{m}\textstyle\sum_{j>\ell} 2^j |\T_j|
= \frac{k^2}{n^2}\,d\cdot \E\!\left[2^X\mathbf 1\{X>\ell\}\right],
\]
where $X$ is hypergeometric with $\Prb[X=j]=|\T_j|/m$.

By Lemma~\ref{lem:moment_bound_high_overlap} (Appendix~\ref{app:sq_bounds}) with $K=kt$,
$
\E\!\left[2^X\mathbf 1\{X>\ell\}\right]\le \sum_{j>\ell}{(2k^2t^2/n)^j}/{j!}.
$
Apply Lemma~\ref{lem:poisson_tail_full} (Appendix~\ref{app:sq_bounds}) with $\lambda=2k^2t^2/n=o(1)$
to bound this tail by $O(\lambda^{\ell+1}/(\ell+1)!)$.
Multiplying by $d=\ell!(n/(k^2t^2))^\ell$ yields an $o((k^2/n^2)2^\ell)$ contribution in the stated regime, so
$\mathrm{AvgCorr}(A)\le \bar\gamma+o(1)$ for every $A$ with $|A|\ge m/d$.
Hence $\mathrm{SDA}\ge d$.
\end{proof}


\noindent We now state the resulting SQ lower bound, which immediately follows.

\begin{theorem}[SQ lower bound]\label{cor:sq_vstat}
Under the hypotheses of Theorem~\ref{thm:sq_sda}, any randomized SQ algorithm that distinguishes an unknown $\mathsf D_S\in\mathcal F$
from $\mathcal D_0$ with success probability at least $2/3$ must make $\Omega(d)$ queries to $\mathrm{VSTAT}(N)$ with $N=\Theta(1/\bar\gamma)$.
In particular, assuming additionally that $kt=\omega((\log n)^2)$ and choosing
$\ell=\lfloor \delta \log_2 n\rfloor$ yields
\vskip-1em
\[
d=\exp(\Omega((\ln n)^2))
\qquad\text{and}\qquad
N=\Theta\!\left({n^{2-\delta}}/{k^2}\right).
\]
\end{theorem}

\begin{remark}\label{rem:partition}
There is an SQ-preserving post-processing $T_P$ (fix a random partition $P$) that maps the \emph{single}-plant row-mixture with $|S|=K$
to a \emph{multi}-plant row-mixture whose blocks are $S\cap P_i$; the block sizes fluctuate and sum to $K$.
This shows the multi-plant problem is no harder than the single case under this transformation.
However, two issues prevent deriving our equal-size lower bound from this reduction.
First, no SQ post-processing (that preserves the null) can, for all $S$, force \emph{equal} block sizes $|S_i|=k$; hence the reduction does not target our equal-size family.
Second, the subset of equal-size plantings aligned to a fixed partition $P$ has size $\prod_i \binom{|P_i|}{k}$, which is a $t^{-\Omega(kt)}$ fraction of all ordered plantings
$\binom{n}{kt}(kt)!/(k!)^t$.
The SDA bound we prove controls average correlation for \emph{every} subset of size at least $m/d$; it does not apply to exponentially thin, partition-aligned slices.
Our proof therefore also rules out algorithms that might succeed on structured subfamilies but cannot handle a constant fraction of all equal-size plantings.
\end{remark}





\bibliographystyle{plainnat}
\bibliography{references}

\appendix
\section{Technical lemmas for Section~\ref{part:sos}}\label{app:sos_tech}

\subsection{Coefficients, calibration, and constraints proofs}\label{app:sos_constraints}

\begin{proof}[proof of Lemma~\ref{lem:nonedge}]
It suffices to check monomials $p=X_{M'}$. Let
\[
M := M'\cup\{(u,r),(v,r)\}
\]
so that $X_M = x_{u,r}x_{v,r}X_{M'}$ and $e=\{u,v\}\in \mathcal{E}_M$.
Consider $F(G)=\wtE_G[X_M]$.

Writing the Fourier expansion of $F$ and splitting terms by whether $e\in T$, one can express
$F(G)=A(G)+G_e B(G)$ where $A,B$ depend only on edges other than $e$.

We claim the Fourier coefficients satisfy $\widehat{F}(T)=\widehat{F}(T\cup\{e\})$ whenever $T$ differs only by toggling $e$.
Indeed, by Definition~\ref{def:trunc} the coefficient $\widehat{F}(T)$ is proportional to
$\E_{(G,X)\sim \mathcal{G}(n,1/2,k,t)}\big[\chi_T(G)\,P(X)\big]$ (and is $0$ if $T$ is truncated).
Since $P(X)$ is an indicator, we may write
\[
\E\big[\chi_T(G)P(X)\big]=\Prb[P(X)=1]\cdot \E\big[\chi_T(G)\mid P(X)=1\big].
\]
On the event $P(X)=1$ we have $x_{u,r}=x_{v,r}=1$, so in the planted model the edge $e$ is deterministically present in $H$ and hence in $G=G_0\cup H$, i.e. $G_e=+1$ almost surely conditional on $P(X)=1$.
Therefore $\chi_{T\cup\{e\}}(G)=\chi_T(G)\cdot G_e=\chi_T(G)$ on this event, implying
$\E[\chi_T(G)\mid P(X)=1]=\E[\chi_{T\cup\{e\}}(G)\mid P(X)=1]$.
Moreover, since $e\in \mathcal{E}_M$, the test graphs $(V,\,T\cup \mathcal{E}_M)$ and $(V,\,(T\cup\{e\})\cup \mathcal{E}_M)$ coincide, so the truncation-feasibility of $T$ and $T\cup\{e\}$ is identical.
This yields $\widehat{F}(T)=\widehat{F}(T\cup\{e\})$.

Hence $A=B$ and $F(G)=A(G)(1+G_e)$, which evaluates to $0$ when $G_e=-1$ (i.e., $e$ is a non-edge).
\end{proof}

\begin{proof}[proof of Lemma~\ref{lem:concentration}]
We argue for $\wtE_G[1]$; the size statistic is analogous.
By Definition~\ref{def:trunc} with $M=\emptyset$,
\[
\wtE_G[1]=\textstyle\sum_{T:\,|V(T)|\le \tau} \widehat{\wtE[1]}(T)\,\chi_T(G).
\]
Under $G\sim G(n,1/2)$ the characters $\{\chi_T\}$ form an orthonormal basis, hence
\[
\E_G[\wtE_G[1]]=\widehat{\wtE[1]}(\emptyset)=1
\qquad\text{and}\qquad
\Var_G(\wtE_G[1])=\textstyle\sum_{T\neq \emptyset} \widehat{\wtE[1]}(T)^2.
\]
Lemma~\ref{lem:coefbd} (with $M=\emptyset$) bounds
$|\widehat{\wtE[1]}(T)| \le \bigl(kt/(n-\tau)\bigr)^{|V(T)|}$ for every truncation-feasible $T$.
Grouping by $s:=|V(T)|$ and using the crude count
$
\#\{T:|V(T)|=s\}\le \binom{n}{s}\,2^{\binom{s}{2}},
$
we obtain
\[
\Var_G(\wtE_G[1])
\le \textstyle\sum_{s=1}^{\tau} \binom{n}{s}\,2^{\binom{s}{2}}\left(\frac{kt}{n-\tau}\right)^{2s}.
\]
Under $kt=n^{1/2-\varepsilon}$ and $\tau\le (\varepsilon/C)\log n$ (for $C$ sufficiently large),
we bound each summand for $1\le s\le \tau$ as follows.
Using $\binom{n}{s}\le n^s$ and $n-\tau\ge n/2$ for large $n$,
\[
\binom{n}{s}\,2^{\binom{s}{2}}\left(\frac{kt}{n-\tau}\right)^{2s}
\le
n^s\cdot 2^{s^2/2}\cdot \left(\frac{2kt}{n}\right)^{2s}
=
2^{s^2/2}\cdot n^s\cdot n^{-(1+2\varepsilon)s}\cdot 2^{2s}
\le
2^{s^2/2+2s}\cdot n^{-2\varepsilon s}.
\]
Now $2^{s^2/2+2s}=\exp\!\left(O(s^2)\right)$, while $n^{-2\varepsilon s}=\exp\!\left(-2\varepsilon s\log n\right)$.
Since $s\le \tau\le (\varepsilon/C)\log n$, we have $s^2=O(\varepsilon^2(\log n)^2)$ and therefore
$O(s^2)\le (\varepsilon/C)\cdot \varepsilon(\log n)^2 \ll \varepsilon s\log n$ for large enough $C$.
Hence each summand is at most $n^{-\Omega(\varepsilon)}$, and summing over $s\le \tau=O(\varepsilon\log n)$ yields
$\Var_G(\wtE_G[1])\le n^{-\Omega(\varepsilon)}$.

For the size statistic $\wtE_G[\sum_{i=1}^n x_{i,j}]$, expand similarly and apply Lemma~\ref{lem:coefbd}
to each monomial $X_M=x_{i,j}$; summing over $i$ changes the coefficient magnitudes only by a polynomial factor,
and the same variance bound yields concentration uniformly over $j\in[t]$ by a union bound.
\end{proof}

\subsection{Separators for ribbons}\label{app:sos_separators}

\begin{proof}[proof of Lemma~\ref{lem:left_right_sep_full}]
This is the standard submodularity/uncrossing argument for minimum vertex cuts in an undirected graph:
the family of minimum separators is closed under intersection/union of appropriate reachable regions,
yielding unique extremal minimum separators. (See, e.g., the ribbon framework in planted-subgraph SoS lower bounds.)
\end{proof}

\subsection{Factorization and PSD proofs}\label{app:sos_factorization}

\begin{proof}[proof of Lemma~\ref{lem:sep_factor}]
Start from Lemma~\ref{lem:ribbon_expansion} and group each feasible ribbon $R$ by its leftmost/rightmost minimum separators
$S_L(R),S_R(R)$, then apply the canonical factorization
$R=(R_\ell,R_m,R_r)$. The contribution that factors as
$\chi_{W_{R_\ell}}\chi_{W_{R_m}}\chi_{W_{R_r}}$ matches $LQL^\top$.
Terms where $R_\ell$ and $R_r$ share vertices outside the separators are exactly the non-disjointness artifacts,
placed into $E_1$.
\end{proof}

\begin{proof}[proof sketch of Lemma~\ref{lem:Q_psd}]
Fix $s$ and restrict $Q$ to indices $S,S'$ with $|S|,|S'|\le s$.
By Definition~\ref{def:LQ}, each entry is a low-degree polynomial in $G$ with Fourier support on edge-sets $W_R$
whose induced ribbon has $|V(R)|\le \tau$.
The standard approach is to bound the spectral norm of the off-diagonal part via the trace-power method:
for an even integer $m$, expand $\E[\mathrm{Tr}((Q^{(\mathrm{off})})^m)]$ as a sum over labeled collections of feasible
ribbons and show that only terms with sufficiently many repeated edges contribute, yielding
$\|Q^{(\mathrm{off})}\|\le \lambda_{\min}(Q^{(\mathrm{diag})})$ with high probability.
The required coefficient input is the ribbon bound from Lemma~\ref{lem:ribbon_expansion},
namely $|\gamma_{S,S'}(R)|\le (kt/(n-\tau))^{|V(R)|}$.
With this input, the enumeration and concentration bounds follow identically to the middle-ribbon PSD proofs in
\cite{BarakHKKMP19} and \cite{MekaPW15}.

In our setting, the only quantitative change is the bound on Fourier coefficients:
by Lemma~\ref{lem:ribbon_expansion}, every feasible middle-ribbon contribution satisfies
\[
|\gamma_{S,S'}(R)| \;\le\; \left(\frac{kt}{n-\tau}\right)^{|V(R)|},
\]
which is the single-plant bound with $k$ replaced by $kt$.
All combinatorial enumeration steps are unchanged because the middle ribbons for $Q$ have no pinned edges
($K=\emptyset$ in Definition~\ref{def:uvk_ribbon}) and separators are defined purely as vertex cuts in the ribbon graph.
\end{proof}


\begin{proof}[proof of Lemma~\ref{lem:iter_sep}]
Iterate Lemma~\ref{lem:sep_factor} on the successive remainder matrices.
At each step, we group ribbons by their (leftmost/rightmost) minimum separators and peel off the contributions with separator size equal to the current layer, placing everything else into the next remainder.
Equivalently (and more explicitly), one may define $E_{m+1}$ to be the contribution to $M$ from feasible ribbons $R$ with $\sep(R)\ge m+1$; then the above iteration yields the stated decomposition with this choice of remainder.
\end{proof}

\begin{proof}[proof of Lemma~\ref{lem:psd_core}]
Apply Lemma~\ref{lem:iter_sep} with $m=\lfloor d/2\rfloor$ to obtain
\[
M = \sum_{s=0}^{\lfloor d/2\rfloor} L^{(s)} Q^{(s)} (L^{(s)})^\top - E_{\lfloor d/2\rfloor+1},
\]
with each $Q^{(s)}\succeq 0$.

We claim $E_{\lfloor d/2\rfloor+1}=0$ identically.
Recall from Lemma~\ref{lem:iter_sep} that $E_{\lfloor d/2\rfloor+1}$ is defined as the total contribution of
truncation-feasible $(A,B)$-ribbons $R$ with
\[
\sep(R)\ge \lfloor d/2\rfloor+1
\]
to the moment matrix entries $M(A,B)$, where $|A|,|B|\le \lfloor d/2\rfloor$.

Fix any such indices $A,B$ and consider any $(A,B)$-ribbon $R$ that can contribute to $M(A,B)$.
By Definition~\ref{def:sep_ab} and our separator convention, the set $Q:=\supp(A)$ is a valid separator,
hence
\[
\sep(R)\le |\supp(A)|\le |A|\le \lfloor d/2\rfloor.
\]
Therefore, no ribbon contributing to any entry indexed by $|A|,|B|\le \lfloor d/2\rfloor$ can satisfy
$\sep(R)\ge \lfloor d/2\rfloor+1$.
Equivalently, the class of ribbons summed into $E_{\lfloor d/2\rfloor+1}$ is empty, and thus
$E_{\lfloor d/2\rfloor+1}=0$.

Thus $M$ is a sum of PSD matrices $L^{(s)}Q^{(s)}(L^{(s)})^\top$, and therefore $M\succeq 0$.
\end{proof}

\section{Technical Lemmas for Section~\ref{part:sq}}

\subsection{SQ proofs}\label{app:sq_proofs}

\begin{proof}[of Lemma~\ref{lem:sq_corr}]
Let $Z_S(x)=\frac1t\sum_{i=1}^t \mathbf 1\{x_v=1\ \forall v\in S_i\}$.
Then $\widehat D_S(x)=p(2^k Z_S(x)-1)$.
Under $\mathcal D_0$, $\Prb[x_v=1\ \forall v\in S_i\cup T_j]=2^{-2k+\lambda_{ij}}$, so
$2^{2k}\langle Z_S,Z_T\rangle=\frac{1}{t^2}\sum_{i,j}2^{\lambda_{ij}}$, yielding the identity.
To upper bound the sum, relax all constraints on $(\lambda_{ij})$ except $\sum_{i,j}\lambda_{ij}=\Lambda$.
Since $2^x$ is convex, the sum is maximized by concentrating all overlap in:
$\sum_{i,j}2^{\lambda_{ij}}\le (t^2-1)\cdot 1+2^\Lambda$.
Substituting gives $p^2(2^\Lambda-1)/t^2\le (k^2/n^2)2^\Lambda$.
\end{proof}

\begin{proof}[of Lemma~\ref{lem:sq_m}]
We count by selecting the  $t $ blocks in order.  For the  $i $–th block (with  $i=1,2,\dots,t $), there are
\[
  \binom{ n - (i-1)k }{k}
\]
ways to choose which  $k $ elements go into that block from the remaining pool.  Thus
\[
  m  =  
  \prod_{i=1}^t \binom{ n - (i-1)k }{k}
   = 
  \binom{n}{k} \binom{n-k}{k} \cdots \binom{n-(t-1)k}{k}.
\]
But each binomial expands as
\[
  \binom{n - (i-1)k}{k}
   = 
  \frac{(n - (i-1)k)!}{(n - ik)! k!},
\]
so the product telescopes:
\[
  m
  = \prod_{i=1}^t
      \frac{(n - (i-1)k)!}{(n - ik)! k!}
  = \frac{n!}{(n - kt)! (k!)^t}.
\]
Finally, since
 $  \binom{n}{kt}
    = \frac{n!}{(kt)! (n-kt)!},
 $
we get the equivalent form
 $
  m  = \binom{n}{kt} \frac{(kt)!}{(k!)^t}.
 $
\end{proof}

\section{Approximations and tail bounds}\label{app:sq_bounds}

This appendix records standard approximations and tail bounds used in the SQ lower bound,
to make Part~\ref{part:sq} fully self-contained.

\subsection{Hypergeometric basics and domination}

\begin{lemma}\label{lem:hypergeom_setup}
Let $U\subseteq[n]$ be fixed with $|U|=K$ and let $V\subseteq[n]$ be uniformly random with $|V|=K$.
Then $X:=|U\cap V|$ is hypergeometric with
\[
\Prb[X=\ell]=\frac{\binom{K}{\ell}\binom{n-K}{K-\ell}}{\binom{n}{K}}.
\]
Moreover, for every integer $j\ge 0$,
\[
\Prb[X\ge j]\le \binom{K}{j}\left(\frac{K}{n}\right)^{\!j}.
\]
\end{lemma}

\begin{proof}
For the tail bound: choose $j$ elements of $U$ that land in $V$.
By union bound over all $\binom{K}{j}$ subsets $J\subseteq U$ of size $j$,
\[
\Prb[X\ge j]\le \sum_{J\subseteq U,|J|=j}\Prb[J\subseteq V].
\]
For fixed $J$, since $V$ is a uniform $K$-subset,
\[
\Prb[J\subseteq V]=\frac{\binom{n-j}{K-j}}{\binom{n}{K}}
=\prod_{r=0}^{j-1}\frac{K-r}{n-r}\le \left(\frac{K}{n}\right)^{\!j}.
\]
\end{proof}

\subsection{Binomial/Poisson approximation in the sparse regime}

\begin{lemma}\label{lem:hypergeom_sparse_asymp}
Let $K=K(n)$ satisfy $K=o(\sqrt{n})$ and let $\ell=O(\log n)$ satisfy $\ell^2=o(K)$.
Let $X$ be as in Lemma~\ref{lem:hypergeom_setup}. Then
\[
\Prb[X=\ell]
=\frac{1}{\ell!}\left(\frac{K^2}{n}\right)^{\!\ell}\,(1+o(1)),
\qquad\text{as }n\to\infty.
\]
\end{lemma}

\begin{proof}
Start from the exact expression:
\[
\Prb[X=\ell]=\frac{\binom{K}{\ell}\binom{n-K}{K-\ell}}{\binom{n}{K}}.
\]
We approximate each factor.

First,
\[
\binom{K}{\ell}=\frac{K^\ell}{\ell!}\prod_{r=0}^{\ell-1}\left(1-\frac{r}{K}\right)
=\frac{K^\ell}{\ell!}\left(1+O\!\left(\frac{\ell^2}{K}\right)\right).
\]

Second, using the sequential-exposure view, the hypergeometric is close to $\mathrm{Bin}(K, K/n)$ when $K=o(\sqrt n)$:
the sampling-without-replacement correction introduces a relative error
$1+O(K^2/n)=1+o(1)$ uniformly for $\ell=O(\log n)$.
Thus
\[
\Prb[X=\ell]
=\binom{K}{\ell}\left(\frac{K}{n}\right)^{\!\ell}\left(1-\frac{K}{n}\right)^{K-\ell}(1+o(1))
=\frac{1}{\ell!}\left(\frac{K^2}{n}\right)^{\!\ell}(1+o(1)),
\]
since $(1-K/n)^{K-\ell}=\exp(-K^2/n+o(1))=1+o(1)$ under $K=o(\sqrt n)$.
\end{proof}

\subsection{A Poisson tail bound}

\begin{lemma}\label{lem:poisson_tail_full}
Let $\lambda\ge 0$ and integer $\ell\ge 0$ with $\lambda<\ell+2$. Then
\[
\sum_{j=\ell+1}^\infty \frac{\lambda^j}{j!}
\le \frac{\lambda^{\ell+1}}{(\ell+1)!}\cdot \frac{1}{1-\frac{\lambda}{\ell+2}}.
\]
\end{lemma}

\begin{proof}
For $j\ge \ell+1$,
\[
\frac{\lambda^{j+1}/(j+1)!}{\lambda^j/j!}=\frac{\lambda}{j+1}\le \frac{\lambda}{\ell+2}.
\]
Hence the tail is dominated by a geometric series with ratio at most $\lambda/(\ell+2)<1$:
\[
\sum_{j=\ell+1}^\infty \frac{\lambda^j}{j!}
\le \frac{\lambda^{\ell+1}}{(\ell+1)!}\sum_{r=0}^\infty \left(\frac{\lambda}{\ell+2}\right)^{\!r}
= \frac{\lambda^{\ell+1}}{(\ell+1)!}\cdot \frac{1}{1-\frac{\lambda}{\ell+2}}.
\]
\end{proof}

\subsection{Derivation of the bound used in Theorem~\ref{thm:sq_sda}}

\begin{lemma}\label{lem:moment_bound_high_overlap}
Let $X$ be hypergeometric as in Lemma~\ref{lem:hypergeom_setup} with parameter $K$.
Then for any integer $\ell\ge 0$,
\[
\E\!\left[2^X\mathbf 1\{X>\ell\}\right]
\le \sum_{j>\ell} \frac{(2K^2/n)^j}{j!}.
\]
\end{lemma}

\begin{proof}
By Lemma~\ref{lem:hypergeom_setup} and $\binom{K}{j}\le K^j/j!$,
\[
\Prb[X\ge j]\le \binom{K}{j}\left(\frac{K}{n}\right)^{\!j}\le \frac{(K^2/n)^j}{j!}.
\]
Using the tail-sum bound,
\[
\E\!\left[2^X\mathbf 1\{X>\ell\}\right]
=\sum_{x>\ell} 2^x \Prb[X=x]
\le \sum_{j>\ell} 2^j \Prb[X\ge j]
\le \sum_{j>\ell} 2^j \cdot \frac{(K^2/n)^j}{j!}
= \sum_{j>\ell} \frac{(2K^2/n)^j}{j!}.
\]
\end{proof}

\end{document}